\documentclass[12pt]{article}

\usepackage{amsmath,amssymb,amsthm,fullpage,charter,verbatim}
\usepackage{pifont,amsfonts,amsmath,amstext,amsthm,amssymb}
\usepackage[english]{babel}
\usepackage[latin1]{inputenc}
\usepackage{graphicx,color}
\usepackage{mathtools}
\usepackage{epic,eepic}
\newtheorem{theorem}{Theorem}
\newtheorem{lemma}{Lemma}

\newcommand{\eps}{\epsilon}

\newcommand{\al}{\alpha} 
\newcommand{\la}{\lambda} 
 
\newcommand{\np }{\textbf{NP\,}} 
\newcommand{\maxlin }{MAX-E3LIN2 problem}

\title{\bf New  Inapproximability Bounds for  TSP \\[1ex]}
\author{
Marek Karpinski\thanks{Dept. of Computer Science and the Hausdorff
    Center for Mathematics, University of Bonn.
    Supported in part by DFG grants and the Hausdorff Grant EXC59-1/2.
    Email:~\texttt{marek@cs.uni-bonn.de}}
\and
Michael Lampis\thanks{
    KTH Royal Institute of Technology.
    Research supported by ERC Grant 226203
    Email:~\texttt{mlampis@kth.se}}
    \and
    Richard Schmied\thanks{Dept. of Computer Science, University of Bonn.
    Work supported by Hausdorff Doctoral Fellowship.
    Email:~\texttt{schmied@cs.uni-bonn.de}}
}

\date{ }

\begin{document}
\maketitle

\begin{abstract}

\noindent In this paper, we study the approximability of the metric Traveling
Salesman Problem (TSP) and prove new explicit inapproximability bounds for 
that problem.  
The best up to now known hardness of approximation bounds were
$185/184$ for the symmetric case (due to Lampis) and $117/116$ for the
asymmetric case (due to Papadimitriou and Vempala). We construct here two new
bounded occurrence CSP
reductions which improve these bounds to $123/122$ and $75/74$, respectively.
The latter bound is the first improvement  in more than a decade 
for the case of the  asymmetric TSP. 
One of our main tools, which may be of independent interest, is a new 
construction of a bounded degree wheel amplifier  used in the proof
of our results.
\end{abstract}

\section{Introduction} 
The \textbf{Traveling Salesman Problem (TSP)} is one of
the best known and most fundamental problems in combinatorial optimization.
Determining how well it can be approximated in polynomial time is therefore a
major open problem, albeit one for which the solution still seems elusive. On
the algorithmic side, the best known efficient approximation algorithm for the
symmetric case is still a 35-year old algorithm due to Christofides~\cite{C76}
which achieves an approximation ratio of $3/2$. However, recently there has
been a string of improved results for the interesting special case of Graphic
TSP, improving the ratio to $7/5$ \cite{GSS11,MS11,M12,SV12}. For the asymmetric case
(ATSP), it is not yet known if a constant-factor approximation is even possible,
with the best known algorithm achieving a ratio of $O(\log n/\log \log n)$
\cite{AGM10}.

Unfortunately, there is still a huge gap between the algorithmic results
mentioned above and the best currently known hardness of approximation results
for TSP and ATSP. For both problems, the known inapproximability thresholds are
small constants ($185/184$ and $117/116$ (cf. \cite{L12,PV06}), respectively). 
 In this paper, we try
to improve this situation somehow by giving modular hardness reductions that slightly
improve the hardness bounds for both problems to  $123/122$ and $75/74$, respectively.
The latter bound is the first, for more than a decade now, improvement 
of Papadimitriou and Vempala bound \cite{PV06} for the ATSP.
The method of our solution differs essentially from that  of \cite{PV06}
and uses some new paradigms of the bounded occurrence optimization which
could be also of independent interest in other applications.  
Similarly to \cite{L12}, the hope is
that the modularity of our construction, which goes through an intermediate
stage of a bounded-occurrence Constraint Satisfaction Problem (CSP), will allow
an easier analysis and simplify future improvements. Indeed, one of the main
new ideas we rely on is a certain new variation of the wheel amplifiers first defined
by Berman and Karpinski \cite{BK01} to establish inapproximability for
3-regular CSPs. This construction, which may be of independent interest, allows
us to establish inapproximability for a 3-regular CSP with a special structure.
This special structure then makes it possible to simulate many of the
constraints in the produced graph essentially ``for free'', without using
gadgets to represent them. Thus, even though for the remaining constraints we
mostly reuse gadgets which have already appeared in the literature, we are
still able to obtain improved bounds.

Let us now recall some of the previous work on the hardness of approximation of
TSP and ATSP. Papadimitriou and Yannakakis~\cite{PY93} were the first to
construct a reduction that, combined with the  PCP Theorem~\cite{ALM98}, gave a
constant inapproximability threshold, though the constant was not more than $1
+ 10^{-6}$ for the TSP with distances either one or two.
Engebretsen~\cite{E03} gave the first explicit approximation lower bound of
$5381/5380$   for the problem.  The inapproximability factor was improved to
$3813/3812$ by B\"ockenhauer and Seibert~\cite{BS00}, who studied the
restricted version of the TSP  with distances one, two and three.
Papadimitriou and Vempala~\cite{PV06} proved that it is \np-hard to approximate
the TSP with a factor better than $220/219$.  Presently, the best known
approximation lower bound is $185/184$  due to Lampis~\cite{L12}.

The important restriction of the TSP, in which we consider instances with
distances between cities  being  values in $\{1, \ldots ,B\}$, is often
referred to as  the $(1,B)$-TSP. The best known efficient approximation
algorithm for the $(1, 2)$-TSP has an approximation ratio $8/7$ and is due to
Berman and Karpinski~\cite{BK06}.  As for  lower bounds, Engebretsen and
Karpinski~\cite{EK06} gave inapproximability thresholds for the $(1,B)$-TSP
problem of $741/740$  for $B = 2$ and $389/388$  for $B = 8$.  More recently,
Karpinski and Schmied~\cite{KS12,KS13} obtained improved inapproximability
factors for  the $(1, 2)$-TSP and the $(1, 4)$-TSP of $535/534$  and $337/336$,
respectively.

For ATSP the currently best known approximation lower bound was $117/116$  due
to Papadimitriou and Vempala~\cite{PV06}.  When we restrict the problem to
distances with values in $\{1, \ldots ,B\}$, there is a simple approximation
algorithm with approximation ratio $B$ that constructs an arbitrary  tour as
solution.  Bl\"aser~\cite{B04} gave an efficient approximation algorithm  for
the $(1,2)$-ATSP with approximation ratio  $5/4$.  Karpinski and
Schmied~\cite{KS12,KS13} proved that it is \np-hard to approximate the $(1,
2)$-ATSP and the $(1, 4)$-ATSP within any factor   less than $207/206$  and
$141/140$, respectively.  For the case $B = 8$, Engebretsen and
Karpinski~\cite{EK06} gave an inapproximability threshold of $135/134$.

\textbf{Overview:} In this paper we give a hardness proof which proceeds in two
steps. First, we start from the MAX-E3-LIN2 problem, in which we are given a
system of linear equations mod $2$ with exactly three variables in each
equation and we want to find an assignment such as to maximize the number of
satisfied equations. Optimal inapproximability results for this problem were
shown by H\aa stad~\cite{H01}. We reduce this problem to a special case where
variables appear exactly $3$ times and the linear equations have a particular
structure. The main tool here is a new variant of the wheel amplifier graphs of
Berman and Karpinski \cite{BK01}.

In the second step, we reduce this $3$-regular CSP to TSP and ATSP. The general
construction is similar in both cases, though of course we use different
gadgets for the two problems. The gadgets we use are mostly variations of
gadgets which have already appeared in previous reductions. Nevertheless, we
manage to obtain an improvement by exploiting the special properties of the
$3$-regular CSP. In particular, we show that it is only necessary to construct
gadgets for roughly one third of the constraints of the CSP instance, while the
remaining constraints are simulated without additional cost using the
consistency properties of our gadgets. This idea may be useful in improving the
efficiency of approximation-hardness reductions for other problems.

Thus, overall we follow an approach unlike that of \cite{PV06}, where the
reduction is performed in one step, and closer to \cite{L12}. The improvement
over \cite{L12} comes mainly from the idea mentioned above, which is made
possible using the new wheel amplifiers, as well as several other tweaks. The
end result is a more economical reduction which improves the bounds for both
TSP and ATSP.

An interesting question may  be whether our techniques can also be used to
derive improved inapproximability results for  variants of the ATSP and TSP
(cf. \cite{EK06},\cite{KS13} and \cite{KS12}) or other graph problems, such as
the Steiner Tree problem.

 %\newpage    

\section{Preliminaries}

In the following, we give some definitions
concerning directed (multi-)graphs and 
 omit the corresponding definitions for 
 undirected (multi-)graphs  if they follow  from
the directed case. Given a directed graph $G=(V(G), E(G))$
and $E'\subseteq E(G)$, for  $e=(x,y)\in E(G)$, we define $V(e  )=\{x,y  \}$
and $V(E') = \bigcup_{e\in E'} V(e)$. For convenience, we abbreviate 
a sequence of edges $(x_1,x_2)$,    $(x_2, x_3), \ldots, (x_{n-1}, x_n)$
by $x_1 \rightarrow  x_2 
\rightarrow x_3 \rightarrow \ldots \rightarrow x_{n-1}\rightarrow x_n$.
In the undirected case, we use sometimes  $x_1 -  x_2 
- x_3 - \ldots - x_{n-1}- x_n$ instead of $\{x_1,x_2\}$,    
$\{x_2, x_3\}, \ldots, \{x_{n-1}, x_n\}$.
Given a directed (multi-)graph $G$, an \emph{Eulerian} cycle in $G$
is a directed cycle that traverses all edges of $G$ exactly once. 
We refer to $G$ as \emph{Eulerian}, if there exists an Eulerian cycle 
in $G$.
For  a 
multiset $E_T$ of directed edges and $v\in V(E_T)$,
we define the outdegree (indegree) of $v$ with respect to $E_T$, 
denoted by $outd_T (v)$ ($ind_T (v)$), 
to be  the number
of edges in $E_T$ that are outgoing of (incoming to) $v$. 
The \emph{balance} of a vertex $v$ with respect to $E_T$ is defined as
$bal_T(v) = ind_T (v) - outd_T (v)$.
In the case of a multiset $E_T$ of undirected edges, 
we define the balance $bal_T(v)$ of a vertex $v\in V(E_T)$ to be one if the number
of incident edges in $E_T$  is odd and zero otherwise. 
We refer to vertices  
$v\in V(E_T)$ with $bal_T(v) = 0$ as \emph{balanced} with respect to $E_T$.
It is well known that 
 a (directed) (multi-)graph $G=(V(G), E(G))$ is Eulerian  
if and only if all
edges  are in the same (weakly) connected component and all vertices $v\in V(G)$ are balanced
with respect to $E(G)$.

Given a multiset of edges $E_T$, we denote by $con_T$ the number of (weakly)
connected components in the graph induced by $E_T$.  A \emph{quasi-tour} $E_T$
in a (directed)  graph $G$ is a multiset of edges from $E(G)$ such that all
vertices are balanced with respect to $E_T$ and $V(E_T) = V(G)$. We refer to a
quasi-tour $E_T$ in $G$ as a \emph{tour} if $con_T=1$. Given a cost function
$w:E(G) \rightarrow \mathbb{R}_+ $, the cost of a quasi-tour $E_T$ in $G$ is
defined by $\sum_{e\in E_T} w(e) + 2(con_T -1) $.

In the  Asymmetric Traveling Salesman problem (ATSP), we are given a
directed graph $G=(V(G), E(G))$ with  positive weights on edges and we want
to find an ordering $v_1, \ldots , v_n$ of the vertices such as to minimize
$ d_G(v_n,v_1) +\sum_{i\in [n-1]}d_G(v_i,v_{i+1})$, 
where $d_G$ denotes the shortest path distance in $G$.

In this paper, we will use the following equivalent reformulation of the ATSP:
Given a directed graph $G$ with  weights on edges, we want to find a tour $E_T$
in $G$, that is, a spanning connected multi-set of edges that balances all
vertices, with minimum cost.

The metric Traveling Salesman problem (TSP) is the special case
of  the ATSP, in which instances are  undirected graphs with positive weights on 
edges.

\section{Bi-Wheel Amplifiers}

In this section, we define the bi-wheel amplifier graphs which will be our main
tool for proving hardness of approximation for a bounded occurrence CSP
with some special properties. Bi-wheel amplifiers are a simple variation of the 
wheel amplifier
graphs given in \cite{BK01}. Let us first recall some definitions (see also
\cite{BK03}).

If $G$ is an undirected graph and $X\subset V(G)$ a set of vertices, we say
that $G$ is a \emph{$\Delta$-regular amplifier} for $X$ if the following conditions
hold:

\begin{itemize}

\item All vertices of $X$ have degree $\Delta-1$ and all vertices of
$V(G)\backslash X$ have degree $\Delta$.

\item For every non-empty subset $U\subset V(G)$, we 
have the condition that $|E(U,V(G)\backslash
U)|\ge \min\{\,|U\cap X|, |(V(G)\backslash U) \cap X|\,\}$, where $E(U,V(G)\backslash U)$ is
the set of edges with exactly one endpoint in $U$.

\end{itemize}

We refer to the set $X$ as the set of \emph{contact} vertices and to
$V(G)\backslash X$ as the set of \emph{checker} vertices. Amplifier graphs are
useful in proving inapproximability for CSPs, in which  every variable appears a
bounded number of times. Here, we will rely on $3$-regular amplifiers. A
probabilistic argument for the existence of such graphs was given in
\cite{BK01}, with the definition of wheel amplifiers.

A wheel amplifier with $2n$ contact vertices is constructed as follows: first
construct a cycle on $14n$ vertices. Number the vertices $1,\ldots,14n$ and
select uniformly at random a perfect matching of the vertices whose number is
not a multiple of 7. The matched vertices will be our checker vertices, and the
rest our contacts. It is easy to see that the degree requirements are
satisfied.

Berman and Karpinski~\cite{BK01} gave  a probabilistic argument to prove 
that with high probability the above construction indeed produces an 
amplifier graph, that is,
all partitions of the sets of vertices give large cuts. Here, we will use a
slight variation of this construction, called a bi-wheel.

A \textbf{bi-wheel amplifier} with $2n$ contact vertices is constructed as follows: first
construct two disjoint cycles, each on $7n$ vertices and number the vertices of
each $1,\ldots,7n$. The contacts will again be the vertices whose number is
a multiple of 7, while the remaining vertices will be checkers. To complete the
construction, select uniformly at random a perfect matching from the checkers
of one cycle to the checkers of the other.

Intuitively, the reason that amplifiers are a suitable tool here is that, given a
CSP instance, we can use a wheel amplifier to replace a variable that appears
$2n$ times with $14n$ new variables (one for each wheel vertex) each of which
appears 3 times. Each appearance of the original variable is represented by a
contact vertex and for each edge of the wheel we add an equality constraint
between the corresponding variables. We can then use the property that all
partitions give large cuts to argue that in an optimal assignment all the new
vertices take the same value.

We  use the bi-wheel amplifier in our construction in a similar way. The
main difference is that while cycle edges will correspond to equality
constraints, matching edges will correspond to inequality constraints. The
contacts of one cycle will represent the positive appearances of the original
variable, and the contacts of the other the negative ones. The reason we do
this is that we can encode inequality constraints more efficiently than
equality with a TSP gadget, while the equality constraints that arise from the
cycles will be encoded in our construction ``for free'' using the consistency
of the inequality gadgets.

Before we apply the  construction however, 
we have to prove that the bi-wheel amplifiers still
have the desired amplification properties.

\begin{theorem}
\label{thm:biwheel}

With high probability, bi-wheels are 3-regular amplifiers.

\end{theorem}

\begin{proof}

Exploiting the similarity between bi-wheels and the standard wheel amplifiers of
\cite{BK01}, we will essentially reuse the proof given there. First, some
definitions: We say that $U$ is a \emph{bad set} if the size of its cut is too small,
violating the second property of amplifiers. We say that it is a \emph{minimal
bad set} if $U$ is bad but removing any vertex from $U$ gives a set that is not
bad.

Recall the strategy of the proof from \cite{BK01}: for each partition of the
vertices into $U$ and $V(G)\backslash U$, they calculate the probability (over the random
matchings) that this partition gives a minimal bad set. Then, they take the sum
of these probabilities over all potentially minimal bad sets and prove that the
sum is at most $\gamma^{-n}$ for some constant $\gamma<1$. It follows by union
bound that with high probability, no set is a minimal bad set and therefore, the
graph is a proper amplifier.

Our first observation is the following: consider a wheel amplifier on $14n$
vertices where, rather than selecting uniformly at random a perfect matching
among the checkers, we select uniformly at random a perfect matching from
checkers with labels in the set $\{1,\ldots,7n-1\}$ to checkers with labels in
the set $\{7n+1,\ldots,14n-1\}$. This graph is almost isomorphic to a bi-wheel.
More specifically, for each bi-wheel, we can obtain a graph of this form by
rewiring two edges, and vice-versa. It easily follows that properties that hold
for this graph, asymptotically with high probability also hold for the bi-wheel.

Thus, we just need to prove that a wheel amplifier still has the amplification
property if, rather than selecting a random perfect matching, we select a
random matching from one half of the checker vertices to the other. We will
show this by proving that, for each set of vertices $S$, the probability that
$S$ is a minimal bad set is roughly the same in both cases. After establishing
this fact, we can simply rely on the proof of \cite{BK01}.

Recall that the wheel has $12n$ checker vertices. Given a set $S$ with $|S|=u$,
what is the probability that exactly $c$ edges have exactly one endpoint in
$S$? In a standard wheel amplifier the probability is
$$ 
P(u,c) = {u \choose c} {12n-u \choose c} \frac
{c!(u-c)!!(12n-u-c)!!}{(12n)!!},
$$
where we denote by $n!!$ the product of all odd natural numbers less than or
equal to $n$, and we assume without loss of generality that $u-c$ is even. Let
us explain this: the probability that exactly $c$ edges cross the cut in this
graph is equal to the number of ways we can choose their endpoints in $S$ and
in its complement, times the number of ways we can match the endpoints, times
the number of matchings of the remaining vertices, divided by the number of
matchings overall.

How much does this probability change if we only allow matchings from one half
of the checkers to the other? Intuitively, we need to consider two
possibilities: one is that $S$ is a balanced set, containing an equal number of
checkers from each side, while the other is that $S$ is unbalanced. It is not
hard to see that if $S$ is unbalanced, then, we can easily establish that the cut
must be large. Thus, the main interesting case is the balanced one (and we will
establish this fact more formally).

Suppose that $|S|=u$ and $S$ contains exactly $u/2$ checkers from each side.
Then the probability that there are exactly $c$ edges crossing the cut is 
$$ 
P'(u,c) = { \frac{u}{2} \choose \frac{c}{2} }^2 { \frac{12n-u}{2} \choose
\frac{c}{2} }^2 \left( \frac{c}{2}\right)!^2
\frac{\left(\frac{u-c}{2}\right)!\left(\frac{12n-u-c}{2}\right)!}{(6n)!} 
$$
Let us explain this. If $S$ is balanced and there are $c$ matching edges with
exactly one endpoint in $S$, then, exactly $c/2$ of them must be incident on a
vertex of $S$ on each side, since the remaining vertices of $S$ must have a
perfect matching. Again, we pick the endpoints on each side, and on the
complement of $S$, select a way to match them, select matchings on the
remaining vertices and divide by the number of possible perfect matchings.

Using Stirling formulas, it is not hard to see that $\left(\frac{n}{2}\right)!^2
= \Theta(n! 2^{-n} \sqrt{n})$. Also $n!! =
\Theta(\left(\frac{n}{2}\right)!2^{n/2})$. It follows that $P'$ is roughly the
same as $P$ in this case, modulo some polynomial factors which are not
significant since the probabilities we are calculating are exponentially small.

Let us now also show that if $S$ is unbalanced, the probability that it is a
minimal bad set is even smaller. First, observe that if $S$ is a minimal bad
set whose cut has $c$ edges, we have  $c\le u/6 $. The reason for this is that
since $S$ is bad, then, $c$ is smaller than the number of contacts in $S$ minus
the number of cycle edges cut. It is not hard to see that, in each fragment,
that is, each subset of $S$ made up of a contiguous part of the cycle, two
cycle edges are cut. Thus, the extra edges we need for the contacts the
fragment contains are at most $1/6$ of its checkers.

Suppose now that $S$ contains $u/2+k$ checkers on one side and
$u/2-k$ checkers on the other. The probability that $c$ matching edges
have one endpoint in $S$ is
$$
P''(u,c,k) = 
{\frac{u}{2}+k\choose \frac{c}{2}+k }
{\frac{u}{2}-k\choose \frac{c}{2}-k }
{\frac{12n-u}{2}-k \choose \frac{c}{2}-k }
{\frac{12n-u}{2}+k \choose \frac{c}{2}+k }
\left( \frac{c}{2}+k\right)!
\left( \frac{c}{2}-k\right)!
\frac{\left(\frac{u-c}{2}\right)!\left(\frac{12n-u-c}{2}\right)!}{(6n)!}
$$
The reasoning is the same as before, except we observe that we need to select
more endpoints on the side where $S$ is larger, since after we remove checkers
matched to outside vertices $S$ must have a perfect matching. Observe that for
$k=0$ this gives $P'$. We will show that for the range of values we care about
$P''$ achieves a maximum for $k=0$, and can thus be upper-bounded by
(essentially) $P$, which is the probability that a set is bad in the standard
amplifier. The rest of the proof follows from the argument given in \cite{BK01}. In
particular, we can assume that $k\le c/2$, since $2k$ edges are cut
with probability 1.
To show that the maximum is achieved for $k=0$, we look at
$P''(u,c,k+1)/P''(u,c,k)$. We will show that this is less than 1. Using
the identity ${n+1\choose k+1}/{n\choose k}=\frac{n+1}{k+1}$, we get
$$
\frac{P''(u,c,k+1)}{P''(u,c,k)} = 
\left( 1 - \frac{2k+1}{\frac{c}{2}+k+1}\right) 
\left( 1 + \frac{2k+1}{\frac{u}{2}-k}\right) 
\left( 1 + \frac{2k+1}{\frac{12n-u}{2}-k}\right) 
$$
Using the fact that $1+x<e^x$, we end up needing to prove the following. 

\begin{equation}
\label{equ:1}
\frac{2k+1}{\frac{c}{2}+k+1} > \frac{2k+1}{\frac{u}{2}-k} +
\frac{2k+1}{\frac{12n-u}{2}-k} 
\end{equation}
Combining  that without loss of
generality $u\le 6n$ holds  with the bounds of $c$ and $k$ we have already mentioned,
the inequality~$(\ref{equ:1})$ is straightforward to establish.
\end{proof}

\section{Hybrid Problem}
\label{sec:hybrid}
By using the bi-wheel amplifier from  the previous section,
we are going to prove hardness of approximation for a bounded occurrence CSP
with very special properties. This particular CSP will be   
 well-suited for constructing a reduction to the TSP given in the next section.

As the starting point of our reduction, 
we  make use of the inapproximability result 
 due to H\aa stad~\cite{H01} for the 
\maxlin, which is defined as follows: 
Given  a
system  $I_1$  of  linear equations mod 2, in which
each equation is of the form $x_i\oplus x_j \oplus x_k = b_{ijk}$ with $b_{ijk}
\in \{0,1\}$, we want to find an assignment to the variables of $I_1$ such as
to maximize the number of satisfied equations.

Let $I_1$ be an instance of the \maxlin~and
 $\{x_i\}^\nu_{i=1}$  the set of variables, 
that appear in 
$I_1$. We   denote by $d(i)$
the number of appearances of $x_i$ in $I_1$. 

\begin{theorem}[H\aa stad \cite{H01}]

For every $\epsilon>0$, there exists a constant  $B_\eps$ such that 
given an instance $I_1$ of the \maxlin~with 
$m$ equations and $\max_{i \in [\nu]} d(i)\leq B_\eps$,  
it is \np-hard to decide whether there is an assignment that 
leaves at most $\eps \cdot m $
 equations unsatisfied, or all assignment leave at least 
$(0.5-\epsilon)m$ equations unsatisfied.

\end{theorem}

Similarly to the work by Berman and Karpinski~\cite{BK99}
(see also \cite{BK01} and \cite{BK03}), 
we will  reduce the number of
occurrences of each variable to $3$. 
For this, we will use  
 our  amplifier construction to create 
 special instances of the Hybrid problem,
 which is defined as follows:
 Given a system $I_2$ of linear equations mod $2$
with either three or two variables in each equation,
 we want to 
find an assignment 
such as to maximize the number of satisfied equations.

In particular,
we are going to prove the following theorem.

\begin{theorem}
\label{thm:hybrid} For every constant $\eps > 0$ and $b\in \{0,1\}$, there
exist instances of the Hybrid problem with  $31m$ equations  such that: $(i)$
Each variable occurs exactly three times.  $(ii)$ $21m$ equations are of the
form $x \oplus y =0$, $9m$ equations are of the form $x \oplus y =1 $ and $m$
equations are of the form $x\oplus y \oplus z = b$ .  $(iii)$ It is \np-hard to
decide whether there is an assignment to the variables that leaves at most
$\eps \cdot m$ equations unsatisfied, or  every assignment to the variables
leaves at least $(0.5- \eps)m$ equations unsatisfied.  \end{theorem}

\begin{proof}
Let  $\epsilon>0$ be a constant  and $I_1$ an instance of the 
\maxlin~with  $\max_{i \in [\nu]} d(i) \leq B_\eps$.
For a fixed $b\in \{0,1\}$, we can flip
some of the  literals such that
 all equations in the instance $I_1$ are 
of the form $x\oplus y \oplus z= b$, where $x, y, z $ are variables 
or negations. 
By constructing three more copies of each equation,
in which all possible pairs of literals appear negated, 
we may assume that each variable occurs 
the same number of times negated as unnegated.

Let us fix a variable $x_i$ in $I_1$.
Then, we 
create $7\cdot d(i)= 2\cdot \al$ new variables 
$Var(i)=\{x^{ui}_j, x^{ni}_j  \}^\al_{j=1}$.
In addition, we construct a  bi-wheel amplifier $W_i$ on $2\cdot \al$ vertices
(that is, a bi-wheel with $d(i)$ contact vertices) with the properties
described in Theorem \ref{thm:biwheel}.  Since $d(i) \leq B_\eps$ is a
constant, this can be accomplished in constant time.  In the remainder, we
refer to \emph{contact} and \emph{checker} variables as the elements in
$Var(i)$, whose corresponding index is a contact and checker vertex in $W_i$,
respectively.  We denote by $M(W_i)\subseteq E(W_i)$ the associated perfect
matching on the set of checker vertices  of $W_i$. In addition, we denote by
$C_{n}(W_i)$ and $C_{u}(W_i)$ the set of edges contained in the first and
second cycle of $W_i$, respectively. 

Let us now define the equations of the corresponding instance
of the Hybrid problem.  
For each edge $\{j,k\} \in M(W_i)$, 
we create  the equation $x^{ui}_j \oplus 
x^{ni}_k = 1$ and refer to equations of this form  as \emph{matching equations}. 
On the other hand, for each edge $\{l,t\}$  in the cycle  $C_{q}(W_i)$
with $q\in \{u,n\}$, we introduce the equation $x^{qi}_l \oplus x^{qi}_t = 0$.
Equations of this form  will be called \emph{cycle equations}.
Finally, we replace the $j$-th unnegated appearance of $x_i$ 
in $I_1$ by the contact variable $x^{ui}_\lambda$ with $\lambda =7 \cdot j $, whereas
the $j$-th negated appearance is replaced by $x^{ni}_{\lambda}$.
The former construction yields  $m$ equations with three variables in 
the instance of the Hybrid problem, which we will denote by $I_2$.   
Notice that each variable appears in exactly $3$ equations in $I_2$.
Clearly, we have  $|I_2|=31m$ equations, thereof   $9m$ matching  equations, 
$21m$ cycle equations and $m$ equations of the form $x\oplus y \oplus z=b$.  

A consistent assignment to $Var(i)$ is an assignment with 
 $x^{ui}_j=b$ and  $x^{ni}_j=(1-b)$ for all $j\in [\al]$, where  $b\in \{0,1\}$. 
A consistent assignment to the variables of $I_2$ 
is an assignment that is consistent to $Var(i)$
for all $i\in [\nu]$.
 By standard
arguments using the amplifier constructed in Theorem~\ref{thm:biwheel}, 
it is possible to convert an assignment to a consistent assignment 
without decreasing the number of satisfied equations and the proof
of Theorem~\ref{thm:hybrid} follows. 
\end{proof}
%\todo

\section{TSP}
This section is devoted to  the proof of the following theorem.

	\begin{theorem}
	\label{thm:tspmain}

It is \np-hard to approximate the TSP to within any constant approximation 
ratio less than $123/122$. 

\end{theorem}			

Let us first sketch the high-level idea of the construction. Starting with an
instance of the Hybrid problem, we will construct a graph, where gadgets
represent the equations. We will design gadgets for equations of size three
(Figure \ref{fig:eqn3}) and for equations of size two corresponding to
\emph{matching} edges of the bi-wheel (Figure \ref{fig:eqn2}). We will not
construct gadgets for the cycle edges of the bi-wheel; instead, the connections
between the matching edge gadgets will be sufficient to encode these extra
constraints. This may seem counter-intuitive at first, but the idea here is
that if the gadgets for the matching edges are used in a consistent way (that
is, the tour enters and exits in the intended way) then it follows that the
tour is using all edges corresponding to one wheel and none from the other.
Thus, if we prove consistency for the matching edge gadgets, we implicitly get
the cycle edges ``for free''. This observation, along with an improved gadget
for size-three equations and the elimination of the variable part of the graph,
are the main sources of improvement over the construction of \cite{L12}.

\subsection{Construction} \label{sec:construction}
We are going to  describe the construction that encodes an instance $I_2$
 of the 
Hybrid problem into an instance of the TSP problem. 
Due to Theorem~\ref{thm:hybrid},
we may assume that the equations with three variables in  $I_2$ 
are all of the form $x\oplus y\oplus z=0 $.

In order to ensure that some edges  are
to be used at least once in any valid tour,  we apply the following simple 
 trick that was already used in the work by Lampis~\cite{L12}:
Let $e$ be an edge
  with weight $w$ that we want to be traversed by every tour. 
	 We remove $e$ and replace it with a path of $L$ edges
	and $L-1$ newly created vertices each of degree two, where we think of $L$ as
	a large constant.  Each of the $L$ edges has weight
 $w/L$ and any tour that fails to traverse at least  two  newly created
 edges is not connected. On the other hand, a
 tour that traverses all but one of those edges
  can be extended by adding two copies of the unused edge increasing  
	the cost of the underlying tour by a negligible value.   
In summary, we may assume that our construction
contains   \emph{forced} edges that need to be traversed 
 at least once by any tour. If $x$ and $y$ are vertices, which are connected by
a forced edge $e$, we write $\{x,y\}_F$ or simply $x-_F y$.  In the following,
we refer to unforced edges $e$ with $w(e)=1$ as \emph{simple}. All unforced
edges in our construction will be simple.

Let us start with the description of the corresponding  graph $G_S$: For each
bi-wheel $W_p$, we will construct the subgraph  $G^p$ of $G_S$. For each vertex
of the bi-wheel, we create a vertex in the graph and for each cycle equation $x
\oplus y =0$, we create a simple edge $\{x, y\}$.  Given a matching equation
between two checkers $x^u_i \oplus x^n_j=1$, we connect the vertices $x^u_i$
and $x^n_j$ with two forced edges $\{x^u_i, x^n_j\}^1_F$ and $\{x^u_i,
x^n_j\}^2_F$.  We have $w(\{x^u_i, x^n_j\}^i_F)=2$ for  each $i\in \{1,2\} $.

Additionally, we create a central vertex $s$ that is connected to gadgets
simulating equations with three variables. Let $x\oplus y\oplus z =0 $ be the
$j$-th equation with three variables in $I_2$.  We now create the graph
$G^{3S}_j$ displayed in Figure~\ref{fig:eqn3}~$(a)$, where the (contact)
vertices for $x,y,z$ have already been constructed in the cycles.  The edges
$\{\gamma^\alpha ,\gamma \}_F$ with $\alpha \in \{r,l\}$ and $\gamma \in
\{x,z,y\}$ are all forced edges with $w(\{\gamma^\alpha ,\gamma \}_F)=1.5$.
Furthermore, we have $w(\{e^\alpha_j, s\}_F)=0.5$ for all $\alpha \in \{r,l\}$.
$\{e^r_j, s\}_F$ and $\{e^l_j, s\}_F$ are both forced edges, whereas all
remaining edges of $G^{3S}_j$ are simple. This is the whole description of
$G_S$.
 
\begin{figure}[h]
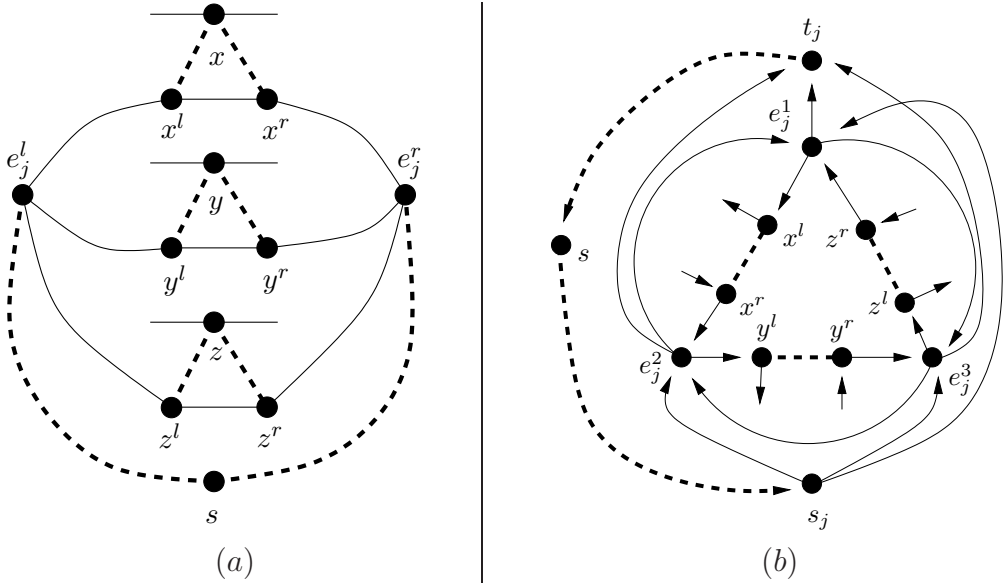

\begin{center}
\begin{tabular}[c]{c|c}
\input{figures/eq3.pspdftex}  &
\input{figures/eqwith3d.pspdftex}\\
$(a)$ & $(b)$
\end{tabular}
%\caption[Gadget for $x\oplus y\olpus z=0$.]
\end{center}
\caption{Gadgets simulating equations with three   variables  in the symmetric case $(a)$ and in the asymmetric case $(b)$. Dotted and straight lines  represent forced and  simple edges, respectively.}
\label{fig:eqn3}
\end{figure}%%%%%%%%%%%%%%%%
%%%%%%%%%%%%%%%%%
%%%%%%%%%%%%%%%
%%%%%%%%%%%%%%%%
%%%%%%%%%%%%%%%%
\subsection{Tour from Assignment}
Given an instance $I_2$ of the Hybrid problem and
an assignment $\phi$ to the variables in $I_2$, we are going to construct
a tour in $G_S$ according to $\phi$ and give the proof of one direction
of the reduction. In particular, we are going to prove the following
lemma. 

\begin{lemma}
\label{lem:tsp1}

If there is an assignment to the variables of a given instance $I_2$ of the
Hybrid problem with $31m$ equations and $\nu$ bi-wheels, that leaves $k$
equations  unsatisfied, then, there exists a tour in $G_S$ with cost at most
$61m+2\nu+k+2$.

\end{lemma}

Before we proceed, let us give a useful definition. Let $G$ be an
edge-weighted graph and $E_T$ a multi-set of edges of $E(G)$ that defines a
quasi-tour. Consider a set $V'\subseteq V(G)$. The local edge cost of the set $V'$
is then defined as 

$$ c_T(V')=\sum_{u\in V'} \sum_{e\in E_T,\  e=\{u,v\}} \frac{w(e)}{2} $$

In words, for each vertex in $V'$, we count half the total weight of its
incident edges used in the quasi-tour (including multiplicities). Observe that
this sum contains half the weight of edges with one endpoint in $V'$ but the
full weight for edges with both endpoints in $V'$ (since we count both
endpoints in the sum). Also note that for two sets $V_1,V_2$, we have
$c_T(V_1\cup V_2)\le c_T(V_1)+c_T(V_2)$ (with equality for disjoint sets) and
that $c_T(V)=\sum_{e\in E_T}w(e)$.

\begin{proof}

First, note that it is sufficient to prove that we can construct a quasi-tour
of the promised cost which uses all forced edges exactly once. Since all
unforced edges have cost 1, if we are given a quasi-tour we can connect two
disconnected components by using an unforced edge that connects them twice
(this is always possible since the underlying graph we constructed is
connected).  This does not increase the cost, since we added two unit-weight
edges and decreased the number of components.  Repeating this results in a
connected tour.

Let $\{W_a\}^\nu_{a=1}$ be the associated set 
of  bi-wheels of $I_2$. For a fixed bi-wheel $W_p$,
let $\{x^u_i,x^n_i\}^{z}_{i=1}$ be its associated set of  variables.
Due to the construction of instances of the Hybrid
problem in Section~\ref{sec:hybrid}, we may assume that all equations with two
variables are satisfied by the given assignment.  Thus, we have $x^u_i \neq
x^n_j $, $x^u_i = x^u_j$ and $ x^n_i = x^n_j$ for all $i,j\in [z]$.

Assuming $x^\al_1=1$ for some $\al\in \{u,n\} $, we use once all simple edges
$\{x^\al_{i}, x^\al_{i+1}\}$ with $i \in [z-1] $ and the edge $\{x^\al_z,
x^\al_{1}\}$. We also use all forced edges corresponding to matching equations
once.  In other words, for each biwheel we select the cycle that corresponds to
the assignment 1 and use all the simple edges from that cycle.  This creates a
component that contains all checker vertices from both cycles and all contacts
from one cycle.

%In order to analyze the cost 
%of a tour $E_T$ locally, we introduce for each pair of checker variables
%$\{x^u_i,x^n_j\}$,
%that appear in a matching equation $x^u_i\oplus x^n_j=1$ in $I_2$, a 
%cost function $c_T(v)$ with $v\in \{x^u_i, x^n_j\}$.  In addition,
% we define a local cost function $c_T(G^p_i)$ for the graph $G^p_i$, 
%where $E(G^p_i)=\{e\in E(G_S)\mid \{x^u_i,x^n_j\}\cap e\neq \emptyset\}$
%and $V(G^p_i)= \bigcup_{e\in E(G^p_i)} e$.
%$$
%c_T(v )=\sum_{e\in E_T \,\colon\, v\in e}w(e)/2~ \textrm{ and }~
%c_T(G^p_i)=c_T(x^u_i)+c_T(x^n_j) + 2 \cdot con_T(G^p_i) ,
%$$
%where $con_T(G^p_i)$ is one if the vertices  $x^u_i$ and  $ x^n_j $ form 
%a connected component
%in the graph induced by $E_T$ and zero, otherwise. 
%  Accordingly,  we obtain $c_T(G^p_i)=5$
%for all checker variables $x^u_i$. For the part that connects $G^p$ with $s$,
%we obtain local cost of $3\cdot 5 +1 =16$.

\begin{figure}[h]
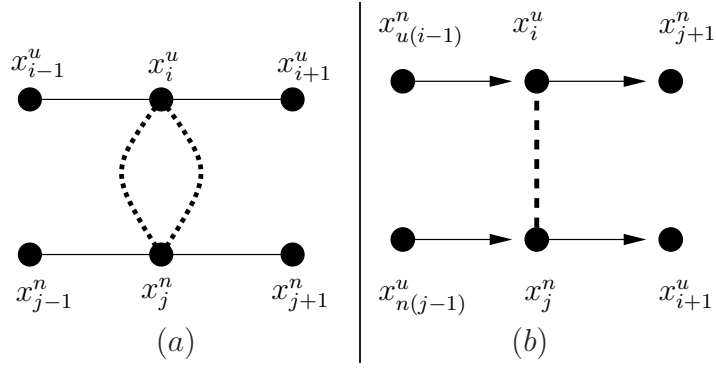

\begin{center}
\begin{tabular}[c]{c|c}
\input{figures/equwith2.pspdftex}  &
\input{figures/eqwith2d.pspdftex}\\
$(a)$ & $(b)$
\end{tabular}
%\caption[Gadget for $x\oplus y\olpus z=0$.]
\end{center}
\caption{Gadget simulating equation  with two variables in symmetric case $(a)$ and in the 
asymmetric case $(b)$.  Dotted and straight lines  represent forced and  simple edges, respectively.}
\label{fig:eqn2}
\end{figure}%%%%%%%%%%%%%%%%

As for the next step, we are going to describe the tour traversing
$G^{3S}_j$ with $j\in [m]$ given an assignment to contact variables. Let us assume that 
$G^{3S}_j$ simulates $x\oplus y \oplus z=0$. According to the assignment 
to $x$, $y$ and $z$, we will traverse $G^{3S}_j$ as follows:
In all cases, we will use all forced edges once.\\
\noindent
\textbf{Case ($x+y+z=2$): }
Then, we use $\{\gamma^l,\gamma^r\}$
 for all   $\alpha\in \{r,l\}$
and $\gamma \in \{x,y,z\}$ with $\gamma=1$. For $\delta \in \{x,z,y\}$ with
$\delta =0$, we use $\{e_j^\alpha,\delta^\al\}$  for all $\alpha\in \{r,l\}$.\\
\textbf{Case ($x+y+z=b$ with $b\in \{0,1\}$):} In both cases, we traverse
$\{\gamma^\alpha,e^\alpha_j\}$  for all $\gamma\in \{x,y,z\}$ and $ \alpha \in
\{r,l  \}$. \\ \textbf{Case ($x+y+z=3$):} We use $\{\gamma^r,\gamma^l\}$ with
$\gamma\in \{y,z\}$. Furthermore, we include $\{x^\alpha ,e^\alpha_j\}$  for
both $\alpha\in \{r,l\}$. 

%\indent
%We are going to define local cost functions for contact variables $\gamma\in \{x,y,z\}$ 
%and  for the graph $G^{3S}_j$.
%$$
%\tilde{
%c}_T(\gamma)= \sum_{\substack{e\in E_T\\
%e \cap \{\gamma^r,\gamma^l\}\neq \emptyset
%}
%}w(e) + \sum_{\substack{e\in E_T \,\colon\,  \gamma\in e\\
%e \not \subseteq  \{\gamma, \gamma^r,\gamma^l\} 
%}
%}
%w(e)/2 ~\textrm{ and } 
%c_T(G^{3S}_j)=\sum_{\gamma\in \{x,y,z\}}c_T(\gamma)
%+
%\sum_{\substack{e\in E_T \\
%e \subseteq  \{s, e^r_j,e^l_j\} 
%}
%}
%w(e)
%$$ 
%Finally, we define $c_T(\gamma)= \tilde{c}_T(\gamma) + 2\cdot con_T(\gamma)$, 
%where $con_T(\gamma)$ is one if the vertices $ \gamma$, $\gamma^l$ and   $\gamma^r$ 
%form a connected component in the graph induced by $E_T$ and zero, otherwise. \\

Let us now analyze the cost of the edges of our quasi-tour given an assignment.
For each matching edge $\{x_i^u,x_j^n\}$ consider the set of vertices made up
of its endpoints. Its local cost is 5: we pay 4 for the forced edges and there
are two used simple edges with one endpoint in the set. Let us also consider
the local cost for a size-three equation gadget, where we consider the set to
contain the contact vertices $\{x,y,z\}$ as well the other 8
vertices of the gadget. The local cost here is 9.5 for the forced edges. We also
pay 6 more (for a total of 15.5) when the assignment satisfies the equation or 7
more when it does not.

Thus, we have given a covering of the vertices of the graph by $9m$ sets of
size two, $m$ sets of size 11 and $\{s\}$. The total edge cost is thus at most $5\cdot
9m+ 15.5\cdot m + 0.5\cdot m + k=61m + k$. To obtain an upper bound on the cost of the
quasi-tour, we observe that the tour has at most $\nu+1$ components (one for
each bi-wheel and one containing $s$). The lemma follows. 

\end{proof}
 
 %%\newpage 
%\begin{figure}[h]
%\begin{center}
%\begin{tabular}[c]{c|c}
%\input{figures/startu.pspdftex}  &
%\input{figures/startd.pspdftex}\\
%$(a)$ & $(b)$
%\end{tabular}
%%\caption[Gadget for $x\oplus y\olpus z=0$.]
%\end{center}
%\caption{Wheel border gadgets  
%in the symmetric case $(a)$ and in the
%asymmetric case $(b)$.}
%\label{fig:start}
%\end{figure}%%%%%%%%%%%%%%%%

\subsection{Assignment from Tour}
In this section, we are going prove the other direction
of our reduction. Given a tour  in $G_S$, we are going to 
define an  assignment to the variables of the associated instance 
of the Hybrid problem and give the proof of the following lemma.  

\begin{lemma}
\label{lem:tsp2}

If there is a tour in $G_S$ with  cost $ 61m +k-2 $, then, there is an
assignment to the variables of the corresponding instance of the Hybrid problem
that leaves at most $k$ equations unsatisfied.

\end{lemma}

Again, let us give a useful definition. Consider a quasi-tour $E_T$ and a set
$V'\subseteq V(G)$. Let $con_T(V')$ be the number of connected components induced
by $E_T$ which are fully contained in $V'$. Then, the full local cost of the
set $V'$ is defined as $ c^F_T(V')= c_T(V')+ 2con_T(V')$. By the definition,
the full local cost of $V(G)$ is equal to the cost of the quasi-tour (plus 2).

Intuitively, $c^F_T(V')$ captures the cost of the quasi-tour restricted to
$V'$: it includes the cost of edges and the cost of added connected components.
Note that now for two disjoint sets $V_1,V_2$ we have $c^F_T(V_1\cup V_2)\ge
c^F_T(V_1)+c^F_T(V_2)$ since $V_1\cup V_2$ could contain more connected
components than $V_1,V_2$ together. If we know that the total cost of the
quasi-tour is small, then $c^F_T(V)$ is small (less than $61m+k$). We can use
this to infer that the sum of the local full costs of all gadgets is small.

The high-level idea of the proof is the following: we will use roughly the same
partition of $V(G)$ into sets as in the proof of Lemma \ref{lem:tsp1}. For each
set, we will give a lower bound on its full local cost for any quasi-tour, which
will be equal to what the tour we constructed in Lemma \ref{lem:tsp1} pays. If
a given quasi-tour behaves differently its local cost will be higher. The
difference between the actual local cost and the lower bound is called the
credit of that part of the graph. We construct an assignment for $I_2$ and show
that the total sum of credits is higher that the number of unsatisfied
equations. But using the reasoning of the previous paragraph, the total sum of
credits will be at most $k$.

\begin{proof} We are going to prove a slightly stronger statement and show that
if there exists a quasi-tour in $G_S$ with cost $ 61m +k-2 $, then, there
exists an assignment leaving at most   $k$ equations  unsatisfied.  Recall that
the existence of a tour in $G_S$ with  cost $C$ implies the existence  of  a
quasi-tour  in $G_S$ with cost at most $C$.

We may assume that simple edges are  contained only once in $E_T$ due to 
the following preprocessing step:
If $E_T$ contains two copies of the same simple edge, we remove them
 without increasing the cost, since the number of components can only increase by one.

In the following, 
given a quasi-tour $E_T$ in $G_S$, we are going to define an 
assignment $\phi_T$   and analyze the number of satisfied equations by $\phi_T$
compared to the cost of the quasi-tour.  

The general idea is that each vertex of $G_S$ that corresponds to a variable of
$I_2$ has exactly two forced and exactly two simple edges incident to it. If
the forced edges are used once each, the variable is called honest. We set it
to 1 if the simple edges are both used once and to 0 otherwise. It is not hard
to see that, because simple cycle edges connect vertices that represent the
variables, this procedure will satisfy all cycle equations involving honest
variables. We then argue that if other equations are unsatisfied the tour is
also paying extra, and the same happens if a variable is dishonest.

Let us give more details. First, we concentrate on the assignment for checker
variables.  

\subsubsection{Assignment for Checker Variables} 

Let us consider the following equations with two variables $x^u_{i-1}\oplus
x^u_{i}=0$, $x^u_{i}\oplus x^u_{i+1}=0$, $x^n_{j-1}\oplus x^n_{j}=0$,
$x^n_{j}\oplus x^n_{j+1}=0$ and $x^u_{i}\oplus x^n_{j}=1$. We are going to
analyze the cost of a quasi-tour traversing the gadget displayed in
Figure~\ref{fig:eqn2}~$(a)$ and define an assignment according to $E_T$. Let us
first assume that our quasi-tour is honest, that is, the underlying  quasi-tour
traverses forced edges only once.

\noindent
\textbf{Honest tours:} 
For  $x\in \{x^u_i,x^n_j\}$, we set $x=1$
 if the quasi-tour traverses both
simple edges incident on $x$ and $x=0$, otherwise. 
Since we removed all copies of the same simple 
 edge, we may assume that cycle equations 
 are always satisfied. 
If  the tour uses 
$x^u_{i-1}-x^u_i -_F x^n_j -_F x^u_i - x^u_{i+1}$, we get
$x^u_{i-1}=x^u_{i+1}=1$, $x^n_{j-1}=x^n_{j+1}=0$ and  $5$ satisfied equations.
Given $x^n_{j-1}-x^n_j -_F  x^u_i-_F x^n_j - x^n_{j+1}$,
we obtain  $5$ satisfied equations as well. Let us define $V^p_i:=\{x^u_i,x^n_j\}$.  
Notice that in both cases, we have
local cost $c^F_T(V^p_i)= 5$. We claim that $c^F_T(V^p_i) \geq 5$ for a valid
quasi-tour. In order to obtain a valid quasi-tour, we need to traverse both
forced edges in $G^p_i$ and use at least two simple edges, as otherwise, it
implies $c^F_T(G^p_i) \geq 6$.  Given a quasi-tour $E_T$, we introduce a local
credit function  defined by $cr_T(V^p_i)= c^F_T(V^p_i)-5$.  If $x^u_i -_F x^n_j
-_F x^u_i $ forms a connected component, we get $4$ satisfied equations and
$cr_T(V^p_i)= 1$, which is sufficient to pay for the unsatisfied equation
$x^u_i\oplus x^n_j=1$.  On the other hand, assuming $x^u_{i-1}=x^u_{i+1}=1$ and
$x^n_{j-1}=x^n_{j+1}=1$, we get $cr_T(V^p_i)=1$ and $1$ unsatisfied equation.\\
\textbf{Dishonest tours:} We are going to analyze quasi-tours,  which are using
one of the forced edges  twice.  By setting $x^u_i\neq x^n_j$, we are able to
find an assignment that always satisfies $x^u_i\oplus x^n_j=1$ and two other
equations out of the five that involve these dishonest variables. The local
cost in this case is at least $7$.  Hence, the credit $cr_T(V^p_i)=2$ is
sufficient to pay for the two unsatisfied equations.

\subsubsection{Assignment for Contact Variables}

Again, we will distinguish between honest tours (which use forced edges exactly
once) and dishonest tours. This time we are interested in seven equations: the
size-three equation $x\oplus y\oplus z=0$ and the six cycle equations
containing the three contacts.

Observe that the local cost of 
$V^{3S}_j:=\{x^r, x^l,x, y^r, y^l,y, z^r, z^l, z,
e_j^r, e_j^l  \}$ is at least 15.5. The local edge cost of any quasi-tour
is 9.5 for the forced edges. For each component $\{\gamma,\gamma^l,\gamma^r\}$
with $\gamma\in\{x,y,z\}$, we need to pay at least 2 more because there are two
vertices with odd degree ($\gamma^l, \gamma^r$) and we also need to connect the
component to the rest of the graph (otherwise the component already costs 2
more). Let us define the credit of $V^{3S}_j$ with respect to $E_T$ 
by $ cr_T=c^F_T(V^{3S}_j)-15.5$.\\
\\
\noindent \textbf{Honest tours:} For each  $\gamma\in \{x,y,z\}$, we set
$\gamma=1$ if the tour uses both simple edges incident on $\gamma$ and $0$,
otherwise.  Notice that in the case  $(x+y+z=b )$ with $b\in
\{0,2\}$, this satisfies all seven equations and the tour has local cost at
least $c^F_T(V^{3S}_j)=15.5$.

Case $(x=y=z=1):$ The assignment now failed to satisfy the size-three equation,
so we need to prove that the quasi-tour has local cost at least 16.5. Since all
vertices are balanced with respect to $E_T$, the quasi-tour has to use at least
one edge incident on $e^r_j$ and $e^l_j$ besides $\{s,e^r_j\}_F$ and
$\{s,e^l_j\}_F$.  If the quasi-tour takes $\{e^\al_j,\gamma^\al\}$ for a $\gamma
\in \{x,y,z\}$ and all $\al \in \{r,l\} $, since all simple edges incident on
$x,y,z$ are used, we get at total cost of at least 16.5, which gives a credit of
1.

Case $(x+y+z=1):$  Without loss of generality, we assume that $x=y=0\neq z$
holds. Again, only the size-three equation is unsatisfied, so we must show that
the local cost is at least $16.5$. We will discuss two subcases.  $(i)$ There is a
connected component $\delta -_F \delta^r - \delta^l -_F \delta $ for some
$\delta \in \{x,y\}$.  We obtain that $c^F_T(\{\delta,\delta^l,\delta^r\}) \geq
6$ and therefore, a lower bound on the total cost of $16.5$.
  $(ii)$ Since we may assume that $x^r$, $x^l$, $y^r$ and $y^l$ are
balanced with respect to $E_T$, we have that $\{e^\al_j,\gamma^\al\}\in E_T$ for
all $\al\in \{r,l\}$ and $\gamma\in \{x,y\}$.  Because $e^\al_j$ are also
balanced, we obtain $\{e^\al_j,z^\al\}\in E_T$ for all $\al \in \{r,l\}$, which
implies a total cost of $16.5$.\\

\noindent \textbf{Dishonest tours:} Let us assume that the quasi-tour uses both
of the forced edges $\{\gamma^r, \gamma\}$ and $\{\gamma^l, \gamma\}$ for some
$\gamma \in \{x,z,y\}$ twice. We delete both copies and add $\{\gamma^r,
\gamma^l\}$ instead which reduces the cost of the quasi-tour.  Hence, we may
assume that only one of the two incident forced edges is used twice. 
 
First, observe that if all forced edges were used once, then there would be
eight vertices in the gadget with odd degree: $x^r, x^l, y^r, y^l, z^r, z^l,
e_j^r, e_j^l$. If exactly one forced edge is used twice, then seven of these
vertices have odd degree. Thus, it is impossible for the tour to make the
degrees of all seven even using only the simple edges that connect them. We can
therefore assume that if a forced edge is used twice, there exists another
forced edge used twice.

We will now take cases, depending on how many of the vertices $x,y,z$ are
incident on forced edges used twice. Note that if one of the forced edges
incident on $x$ is used twice, then exactly one of the simple edges incident on
$x$ is used once. So, first suppose all three of $x,y,z$ have forced edges used
twice. The local cost from forced edges is at least $14$. Furthermore, there
are three vertices of the form $\gamma^\al$, for $\gamma\in\{x,y,z\}$ and
$\al\in\{l,r\}$ with odd degree. These have no simple edges connecting them,
thus the quasi-tour will use three simple edges to balance their degrees.
Finally, the used simple edges incident on $x,y,z$ each contribute $0.5$ to the
local cost. Thus, the total local cost is at least $18.5$, giving us a credit of
$3$. It is not hard to see that there is always an assignment satisfying four
out of the seven affected equations, so this case is done.

Second, suppose exactly two of $x,y,z$ have incident forced edges used twice,
say, $x,y$. For $z$, we select the honest assignment (1 if the incident simple
edges are used, 0 otherwise) and this satisfies the cycle equations for this
variable. We can select assignments for $x,y$ that satisfy three of the
remaining five equations, so we need to show that the cost in this case is at
least $17.5$. The cost of forced edges is at least $12.5$, and the cost of simple
edges incident on $x,y$ adds $1$ to the local cost. One of the vertices
$x^l,x^r$ and one of $y^l,y^r$ have odd degree, therefore the cost uses two
simple edges to balance them. Finally, the vertices $z^l,z^r$ have odd degree.
If two simple edges incident to them are used, we have a total local cost of 17.5.
If the edge connecting them is used, then the two simple edges incident on $z$
must be used, again pushing the local cost to 17.5.

Finally, suppose only $x$ has an incident forced edge used twice. By the parity
argument given above, this means that one of the forced edges incident on $s$ is
used twice. We can satisfy the cycle equations for $y,z$ by giving them their
honest assignment, and out of the three remaining equations some assignment to
$x$ satisfies two. Therefore, we need to show that the cost is at least 16.5. The
local cost from forced edges is $11.25$ and the simple edge incident on $x$
contributes $0.5$. Also, at least one simple edge incident on $x^l$ or $x^r$ is
used, since one of them has odd degree. For $y^l,y^r$, either two simple edges
are used, or if the edge connecting them is used the simple edges incident on
$y$ contribute $1$ more. With similar reasoning for $z^l,z^r$, we get that the
total local cost is at least $16.75$.\\

\noindent
Let us now conclude our analysis. Consider the following partition of $V$: we
have a singleton set $\{s\}$, $9m$ sets of size $2$ containing the matching
edge gadgets and $m$ sets of size $11$ containing the gadgets for size-three
equations (except $s$). The sum of their local costs is at most $c^F_T(V)\le
61m +k$. But the sum of their local costs is (using the preceding analysis)
equal to $61m+\sum cr_T(V_i)$. Thus, the sum of all credits is at most $k$. Since
we have already argued that the sum of all credits is enough to cover all
equations unsatisfied by our assignment, this concludes the proof.

 \end{proof}

%\begin{tabular}{|c|c|c|c|}
%\hline
%& & &  \\
%Case & costs & credit & number of unsatisfied equations\\
%\hline
%$(1,1)$ & $\geq 7.5$ & $\geq 1.5$ &  $\leq 1$ \\
%$(3,1)$ & $\geq 9.5$ & $\geq 3.5$ & $\leq 3$ \\
%$(0,2)$ & $\geq 7$ & $\geq 1$ & $\leq 1$  \\
%$(2,2)$ & $\geq 9$ & $\geq 3$ & $\leq 2$\\
%%& & &  \\ 
%\hline
%\end{tabular}\\
%\mbox{}\\*
%\end{centering}
%\vspace{0.5cm}
%

We are ready to give the proof of  Theorem~\ref{thm:tspmain}.      
\begin{proof}[Proof of Theorem~\ref{thm:tspmain}.]

We are given an instance $I_1$ of the \maxlin~with $\nu$ variables and $m$
equations. For all $\delta>0$, there exists a $k$ such that if we repeat each
equation $k$ time we get an instance $I^{(k)}_1$ with $m'=km$ equations and
$\nu$ variables such that $2(\nu+1)/m'\le \delta$.

Then, from $I^{(k)}_1$, we generate an instance $I_2$ of the Hybrid problem and
the corresponding graph $G_S$. Due to Lemmata~\ref{lem:tsp1},~\ref{lem:tsp2}
and Theorem~\ref{thm:hybrid}, we know that for all $\eps>0$, it is \np-hard to
tell whether there is a tour with cost at most $61m'+ 2\nu +2 +\eps \cdot m'
\leq 61\cdot m' +(\delta+\eps)m'$ or all tours have cost at least $61m' + (0.5
-\eps) m' - 2 \geq 61.5 \cdot m' -\eps\cdot m' -\delta \cdot m' $. The ratio between these
two cases can get arbitrarily close to $123/122$ by appropriate choices for
$\epsilon, \delta$.

  \end{proof}

\section{ATSP} In this section, we  prove the following theorem.

\begin{theorem} \label{thm:atspmain} 

It is \np-hard to approximate the ATSP to within any constant approximation
ratio less than $75/74$. 

\end{theorem}

\subsection{Construction}

Let us  describe the construction that encodes an instance $I_2$ of the Hybrid
problem into an instance of the ATSP. Again, it will be useful to have the
ability to force some edges to be used, that is, we would like to have bidirected
forced edges. A bidirected forced edge of weight $w$ between two vertices $x$
and $y$ will be created in a similar way as undirected forced edges in the
previous section: construct $L-1$ new vertices and connect $x$ to $y$ through
these new vertices, making a bidirected path with all edges having weight
$w/L$. It is not hard to see that without loss of generality we may assume that
all edges of the path are used in at least one direction, though we should note
that the direction is not prescribed. In the remainder, we denote a directed
forced edge consisting of vertices $x$ and $y$ by $(x,y)_F$, or $x\rightarrow_F
y$. 

Let $I_2$ consist of the collection $\{W_i\}^\nu_{i=1}$ of bi-wheels. Recall
that the bi-wheel consists of two cycles and a perfect matching between their
checkers. Let $\{x^u_i, x^n_i\}^z_{i=1}$ be the associated set of variables of
$W_p$. We write $u(i)$ to denote the function which, given the index of a
checker variable $x^u_i$ returns the index $j$ of the checker variable $x^n_j$
to which it is matched (that is, the function $u$ is a permutation function
encoding the matching). We write $n(i)$ to denote the inverse function
$u^{-1}(i)$.

Now, for each bi-wheel $W_p$, we are going to construct the corresponding
directed graph $G^p_A$ as follows. First, construct a vertex for each checker
variable of the wheel. For each matching equation $x^u_i\oplus x^n_j =1$, we
create a bidirected forced edge $\{x^u_i, x^n_j\}_F$ with $w(\{x^u_i,
x^n_j\}_F)=2$.

For each contact variable $x_k$, we create two corresponding vertices $x^r_k$
and $x^l_k$, which are joined by the bidirected forced edge $\{x^r_k,
x^l_k\}_F$ with $w(\{x^r_k, x^l_k\}_F)=1$.  

Next, we will construct two directed cycles $C^p_u$ and $C^p_n$. Note that we
are doing arithmetic on the cycle indices here, so the index $z+1$ should be
read as equal to $1$. For $C^p_u$, for any two consecutive checker vertices
$x^u_i, x^u_{i+1}$ on the un-negated side of the bi-wheel, we add a simple
directed edge $x^n_{u(i)}\to x^u_{i+1}$.  If the checker $x^u_i$ is followed by
a contact $x^u_{i+1}$ in the cycle, then we add two simple directed edges
$x^n_{u(i)}\to x^{ur}_{i+1}$ and $x^{ul}_{i+1}\to x^u_{i+2}$.  Observe that by
traversing the simple edges we have just added, the forced matching edges in
the direction $x^u_i\to_F x^n_{u(i)}$ and the forced contact edges for the
un-negated part in the direction $x^{ur}_i\to_F x^{ul}_i$ we obtain a cycle that
covers all checkers and all the contacts of the un-negated part.

We now add simple edges to create a second cycle $C^p_n$. This cycle will
require using the forced matching edges in the opposite direction and, thus,
truth assignments will be encoded by the direction of traversal of these edges.
First, for any two consecutive checker vertices $x^n_i, x^n_{i+1}$ on the
un-negated side of the bi-wheel, we add the simple directed edge $x^u_{n(i)}\to
x^n_{i+1}$.  Then, if the checker $x^n_i$ is followed by a contact $x^n_{i+1}$
in the cycle then we add the simple directed edges $x^u_{n(i)}\to x^{nr}_{i+1}$
and $x^{nl}_{i+1}\to x^n_{i+2}$. Now by traversing the edges we have just
added, the forced matching edges in the direction $x^n_i\to_F x^u_{n(i)}$ and the
forced contact edges for the negated part in the direction $x^{nr}_i\to_F
x^{nl}_i$, we obtain a cycle that covers all checkers and all the contacts of
the negated part, that is, a cycle of direction opposite to $C^p_u$.

What is left is to encode the equations of size three.  Again, we have a
central vertex $s$ that is connected to gadgets simulating equations with three
variables. For every equation with three variables, we create the  gadget
displayed in Figure~\ref{fig:eqn3} $(b)$, which is  a variant of the gadget
used by Papadimitriou and Vempala~\cite{PV06}.  Let us assume that the $j$-th
equation with three variables in $I_3$ is of the form $x \oplus y \oplus z =
1$.  This equation is simulated by $G^{3A}_j$. The vertices used are the
contact vertices $\gamma^\al, \gamma\in \{x,y,z\},\al\in \{r,l\}$, which we
have already introduced, as well as the vertices  $\{s_j, t_j, e^i_j \mid i\in
[3]  \}$. For notational simplicity, we define 
$V^{3A}_j=\big\{s_j, t_j, e^i_j, \gamma^\al \mid i\in [3], \gamma\in \{x,y,z\},
\al\in \{r,l\} \big\}$.
  All directed non-forced edges are simple.  The vertices $s_j$ and
$t_j$ are connected to $s$ by forced edges with
$w((s,s_j)_F)=w((t_j,s)_F)=\lambda$, where $\lambda>0$ is a  small fixed
constant. To simplify things, we also force them to be used in the displayed
direction by deleting the edges that make up the path of the opposite
direction.  This is the whole description of the graph $G_A$.

 \subsection{Assignment to Tour}
We are going to construct a tour in $G_A$ given an assignment to the variables of 
$I_2$ and prove the following lemma. 
\begin{lemma}
\label{lem:atsp1}

 Given an instance $I_2$ of the Hybrid problem with $\nu$ bi-wheels and an 
assignment that leaves $k$ equations in $I_2$ unsatisfied,  then, there  exists
a tour in $G_A$ with cost at most $37m+5\nu +2m \lambda + 2\nu \lambda+k$.

 \end{lemma}  

Before we proceed, let us again give a definition for a local edge cost
function. Let $G$ be an edge-weighted digraph and $E_T$ a multi-set of
edges of $E(G)$ that defines a tour. Consider a set $V'\subseteq V(G)$. The local
edge cost of the set $V'$ is then defined as 

$$ c_T(V')=\sum_{u\in V'} \sum_{(u, v)\in E_T} w\big((u, v) \big) $$

In words, for each vertex in $V'$ we count the total weight of its outgoing
edges used in the quasi-tour (including multiplicities). Thus, that this sum
contains the full weight for edges with their source in $V'$, regardless of
where their other endpoint is.  Also note that again for two sets $V_1,V_2$ we
have $c_T(V_1\cup V_2)\le c_T(V_1)+c_T(V_2)$ (with equality for disjoint sets)
and that $c_T(V)=\sum_{e\in E_T}w(e)$.

\begin{proof}[Proof of Lemma~\ref{lem:atsp1}]

Let $W_p$ be a bi-wheel with variables $\{x^u_i,x^n_i\}^z_{i=1 }$.  Given an
assignment to the variables of $I_2$, due to Theorem~\ref{thm:hybrid}, we may
assume that either $x^u_i=1\neq x^n_j$ for all $i,j\in [z]$ or $x^u_i=0\neq
x^n_j$ for all $i,j\in [z]$.  We traverse the  cycle $C^p_u$ if $x^u_1=1$ and
the  cycle $C^p_n$ otherwise. This creates $\nu$ strongly connected components.
Each contains all the checkers of a bi-wheel and the contacts from one side.

For each matching edge gadget, the local edge cost is $3$. We pay two for the
forced edge and $1$ for the outgoing simple edge. We will account for the cost
of edges incident on contacts when we analyze the size-three equation gadget
below.

Let us describe the part of the tour traversing the graph $G^{3A}_j$, which
simulates $x\oplus y \oplus z= 1$. Recall that if $x$ is set to true in the
assignment we have traversed the bi-wheel gadgets in such a way that the forced
edge $x^r\to_F x^l$ is used, and the simple edge coming out of $x^l$ is used.

According to the assignment to $x$, $y$ and $z$, we traverse $G^{3A}_j$ as
follows: 

\textbf{ Case $(x+y+z=1)$:} Let us assume that $z=y=0\neq x$ holds.  Then, we
use $s\rightarrow_F s_j\rightarrow e^2_j \rightarrow y^l\rightarrow_F
y^r\rightarrow e^3_j \rightarrow z^l \rightarrow_F z^r \rightarrow e^1_j
\rightarrow t_j \rightarrow_F s $. The cost is $3+\lambda$ for the forced
edges, $6$ for the simple edges inside the gadget, plus $1$ for the simple edge
going out of $x^l$. Total local edge cost cost: $c_T(V^{3A}_j) =10+\lambda$.

\textbf{ Case $(x+y+z=3)$:} Then, we use $s\rightarrow_F s_j\rightarrow e^2_j
\rightarrow e^1_j \rightarrow e^3_j \rightarrow t_j \rightarrow_F  s $. Again
we pay $3+\lambda$ for the forced edges, $4$ for the simple edges inside the
gadget and $3$ for the outgoing edges incident on $x^l,y^l,z^l$. Total local
edge cost: $c_T(V^{3A}_j) = \lambda+10$. 

\textbf{ Case $(x+y+z=2)$: }Let us assume that $x=y=1\neq z$ holds.  Then, we
use $s\rightarrow_F s_j\rightarrow e^3_j \rightarrow z^l \rightarrow_F z^r
\rightarrow e^1_j \rightarrow e^3_j \rightarrow e^2_j \rightarrow t_j
\rightarrow_F s $ with total local edge cost  $c_T(V^{3A}_j) = \lambda+11$.

\textbf{ Case $(x+y+z=0)$: } We use $s\rightarrow_F s_j\rightarrow e^2_j
\rightarrow y^l\rightarrow_F y^r\rightarrow e^3_j \rightarrow z^l \rightarrow_F
z^r \rightarrow e^1_j \rightarrow x^l \rightarrow_F x^r \rightarrow e^2_j
\rightarrow t_j \rightarrow_F s $ with $c_T(V^{3A}_j) = \lambda+11$. 

The total edge cost of the quasi-tour we constructed is $3\cdot 9m +
(10+2\lambda)m + k= 37m + 2\lambda m + k$. We have at most $\nu +1$ strongly
connected components: one for each bi-wheel and one containing $s$. A component
representing a bi-wheel can be connected to $s$ as follows: let $x^l,x^r$ be
two contact vertices in the component. Add one copy of each edge from the cycle
$s\to_F s_j \to e^1_j \to x^l \to_F  x^r \to e^2_j \to t_j \to_F s$. This
increases the cost by $5+2\lambda$ but decreases the number of components by
one.

\end{proof}

\subsection{Tour to Assignment}
In this section, we are going to prove the other direction 
of the reduction.

 \begin{lemma}
\label{lem:atsp2}
 
If there is a tour with cost $37\cdot m + k+ 2\lambda \cdot m$, then, 
there is an assignment that 
leaves at most $k$ equations unsatisfied.

 \end{lemma}   

\begin{proof}

Given a tour $E_T$ in $G_A$, we are going to define an assignment to checker
and contact variables. As in Lemma \ref{lem:tsp2}, we will show that any tour
must locally spend on each gadget at least the same amount as the tour we
constructed in Lemma \ref{lem:atsp1}. If the tour spends more, we use that
credit to satisfy possible unsatisfied equations.

\subsubsection{Assignment for Checker Variables} 

Let us consider the following equations with two variables $x^u_{i} \oplus
x^u_{i+1}=0$, $x^u_{i-1} \oplus x^u_{i}=0$, $ x^u_{i} \oplus x^n_{j}=1$,
$x^n_{j} \oplus x^n_{j+1}=0$, $x^n_{j-1} \oplus x^n_{j}=0$ and the
corresponding situation displayed in Figure~\ref{fig:eqn2} $(b)$.  Since $E_T$
is a valid tour in $G_A$, we know that $\{x^u_i,x^n_j\}_F$ is traversed and due
to the degree condition, for each $x\in \{x^u_i,x^n_j\}$, the tour uses another
incident edge $e$ on  $x$ with $w(e)\geq 1$.  Therefore, we have that
$c_T(\{x^u_i,x^n_j\})\geq 3$. The credit assigned to a gadget is defined as
$cr_T(\{x^u_i,x^n_j\})=c_T(\{x^u_i,x^n_j\})-3$. 

Let us define the assignment for $x^u_i$ and $x^n_j$. A variable $x^u_i$ is
honestly traversed if either both the simple edge going into $x^u_i$ is used
and the simple edge coming out of $x^n_j$ is used, or neither of these two
edges is used. In the first case, we set $x^u_i$ to $1$, otherwise to $0$.
Similarly, $x^n_j$ is honest if both the edge going into $x^n_j$ and the edge
out of $x^u_i$ are used, and we set it to $1$ in the first case and $0$
otherwise.

\textbf{Honest tours:} First, suppose that both $x^u_i$ and $x^n_j$ are honest.
We need to show that the credit is at least as high as the number of
unsatisfied equations out of the five equations that contain them. It is not
hard to see that if we have set $x^u_i\neq x^n_j$ all equations are satisfied.
If we have set both to $1$, then the forced edge must be used twice, making the
local edge cost at least $6$, giving a credit of $3$, which is more than
sufficient.

\textbf{Dishonest tours:} If both $x^u_i$ and $x^n_j$ are dishonest the tour
must be using the forced edge in both directions. Thus, the local cost is $5$
or more, giving a credit of $2$. There is always an assignment that satisfies
three out of the five equations, so this case is done. If one of them is
dishonest, the other must be set to $1$ to ensure strong connectivity. Thus,
there are two simple edges used leaving the gadget, making the local cost $4$
(perhaps the same edge is used twice). We can set the honest variable to $1$
(satisfying its two cycle equations), and the other to $0$, leaving at most one
equation unsatisfied.

\subsubsection{Assignment for Contact Variables} 

First, we note that for any valid tour, we have $c_T(V^{3A}_j)\geq 10+\la$.
This is because the two forced edges of weight $\lambda$ must be used, and
there exist $10$ vertices in the gadget for which all outgoing edges have
weight $1$. Let us define the credit $cr_T(V^{3A}_j)= c_T(V^{3A}_j) -(10+\la)$.

\textbf{Honest Traversals:}  We assume that the underlying tour is honest, that
is, forced edges are traversed  only in one direction. We set $x$ to $1$ if the
forced edge is used in the direction $x^r\to_F x^l$ and $0$ otherwise. In the
first case we know that the simple edges going into $x^r$ and out of $x^l$ are
used. In the second, the edges $e^1_j\to x^l$ and $x^r\to e^2_j$ are used. We
do similarly for $y,z$.

We are interested in the equation $x\oplus y\oplus z=1$ and the six cycle
equations involving $x,y,z$. The assignment we pick for honest variables
satisfies the cycle equations, so if it also satisfies the size-three equation
we are done. If not, we have to prove that the tour pays at least $11+\lambda$.

\textbf{Case $(x=y=z=0)$:} Due to our assumption, we know that $e^2_j
\rightarrow y^l\rightarrow_F y^r\rightarrow e^3_j \rightarrow z^l \rightarrow_F
z^r \rightarrow e^1_j \rightarrow x^l\rightarrow_F x^r\rightarrow e^2_j$ is a
part of the tour.  Since $E_T$ is a tour, there exists a vertex in
$V^{3A}_j\backslash \{s_j,t_j\}$  that is visited twice and we get
$c_T(V^{3A}_j)\geq 11+\lambda $.  Thus, we can spend the credit
$cr_T(V^{3A}_j)\geq 1$ on the unsatisfied equation $x\oplus y\oplus z=1$.

\textbf{ Case $(x+y+z=2)$:} Without loss of generality, let us assume that
$x=y=1\neq z$ holds.  Then, we know that $e^3_	j\rightarrow z^l \rightarrow_F
z^r \rightarrow e^1_j$ is a part of the tour. But, this implies that there is a
vertex in $V(G^{3A}_j)$ that is visited twice. Hence, we have that
$cr_T(V^{3A}_j) \geq 1$. 

%\\ \indent Suppose that the tour uses the forced edge $(s,s_j)_F$ twice.  Since
%all vertices are balanced with respect to $E_T$, it implies that we have to use
%$(t_j, s)_F$ twice as well. Note that $c_T(t_j)\geq 1+\la$ and $c_T(s_j)\geq
%1+\la$ yields $cr_T(G^{3A}_j) \geq 1$ which is sufficient to pay for the
%potentially unsatisfied equation $x\oplus  y \oplus z=1 $.  More generally, if
%we use  $(s,s_j)_F$ and  $(t_j, s)_F$ both $i$-times, we get
%$cr_T(G^{3A}_j)\geq  (i-1).$\\ 

\textbf{Dishonest Traversals:} Consider the situation, in which some forced
edges $\{\gamma^r, \gamma^l\}_F$ are traversed in both directions for some
variables $\gamma \in \{x,y,z\}$. For the honest variables, we set them to the
appropriate value as before, and this satisfies their cycle equations. Observe
now that if a forced edge $\gamma^l \to_F \gamma^r$ is also used in the opposite
direction, then there must be another edge used to leave the set
$\{\gamma^l,\gamma^r\}$. Thus the local edge cost of this set is at least $3$.
It follows that the credit we have for the gadget is at least as large as the
number of dishonest variables. We can give appropriate values to them so each
satisfies one cycle equation and the size-three equation is satisfied. Thus,  the
number of unsatisfied equations is not larger than our credit.

In summary, for every tour $E_T$ in $G_A$, we can find an assignment to the
variables of $I_2$ such that all unsatisfied  equations  are paid by the credit
induced by $E_T$.
 
\end{proof}

We are ready to give the proof of  Theorem~\ref{thm:atspmain}.

\begin{proof}[Proof of Theorem~\ref{thm:atspmain}]

We are again given an instance $I_1$ of the \maxlin~with $\nu$ variables and
$m$ equations. For all $\delta>0$, there exists a $k$ such that if we repeat
each equation $k$ time we get an instance $I^{(k)}_1$ with $m'=km$ equations
and $\nu$ variables such that $\nu/m'\le \delta$.

Then, from $I^{(k)}_1$, we generate an instance $I_2$ of the Hybrid problem and
the corresponding directed graph $G_A$. Due to
Lemmata~\ref{lem:atsp1},~\ref{lem:atsp2} and Theorem~\ref{thm:hybrid}, we know
that for all $\eps>0$, it is \np-hard to tell whether there is a tour with cost
at most $37m'+ 5\nu +2m(\nu+\lambda) +\eps \cdot m' \leq 37\cdot m' +\eps'm'$
or all tours have cost at least $37m' + (0.5 -\eps) m'  \geq 37.5 \cdot m'
-\eps'\cdot m'  $, for some $\eps'$ depending only on $\eps,\delta,\lambda$.
The ratio between these two cases can get arbitrarily close to $75/74$ by
appropriate choices for $\epsilon, \delta, \lambda$.

\end{proof}

\section{ Concluding Remarks } 

In this paper, we proved that it is hard to approximate the  ATSP and the TSP
within any constant factor less than $75/74$  and $123/122$, respectively.
The proof method required essentially new ideas and constructions from the 
ones used before in that context.
Since the best known upper
bound on the approximability is $O(\log n/ \log \log n)$ for ATSP and $3/2$ for TSP, 
there is certainly
room for improvements. Especially, in the asymmetric version of the TSP, there is
a large gap between the approximation lower and upper bound, and it remains a major
open problem on the existence of an efficient constant factor approximation algorithm for
that problem.
Furthermore, it would be nice to investigate if some of the ideas of this
paper, and in particular  the bi-wheel amplifiers, can be used to offer improved
hardness results for other optimization problems, such as the Steiner Tree problem.

\end{document}